\documentclass[twocolumn,notitlepage,pra,superscriptaddress]{revtex4-1}
\usepackage{amsmath}
\usepackage{amsfonts}
\usepackage{amssymb}
\usepackage{graphicx}
\usepackage{sidecap}
\usepackage{braket}

\usepackage{blindtext}
\usepackage{subfigure}
\usepackage{lipsum}
\usepackage{fullpage}

\newcommand*\diff{\mathop{}\!\mathrm{d}}
\usepackage{appendix}
\usepackage{subfigure}
\usepackage{color}
\usepackage[dvipsnames]{xcolor}
\usepackage{physics}

 \usepackage{amsthm}
\newtheorem{claim}{Claim}
\newtheorem{theorem}{Theorem}

\begin{document}
\title{Optimal storage time for $N$ qubits coupled to a one-dimensional waveguide}

\author{Tevfik Can Yuce}
\affiliation{Department of Electrical and Electronics Engineering, Koc University, Istanbul, 34450, Turkey}
\author{Fatih Dinc}
\thanks{Email: fdinc@stanford.edu}
\affiliation{Department of Applied Physics, Stanford University, Stanford, CA, USA}
\author{Agata M. Bra\'nczyk}
\affiliation{Perimeter Institute for Theoretical Physics, Waterloo, Ontario, N2L 2Y5, Canada}

\begin{abstract}
Symmetry-protected subradiance is known to guarantee high qubit storage times in free space. We show that in one-dimensional waveguides, this is also true, but that even longer qubit storage times can be identified by considering the eigenspectrum of the qubit-qubit coupling matrix. In the process, we introduce three theorems about $N$ qubits coupled to a one-dimensional waveguide: i) the coupling matrix, which is otherwise non-singular over a continuum of qubit separation values, contains point-singularities; ii) the collective decay rates have symmetric properties, and iii) a linear chain of qubits coupled to a one-dimensional waveguide exhibits symmetry-protected subradiance. Our results will be beneficial for designing memory applications for future quantum technologies.
\end{abstract}

\maketitle

Quantum memory is a key component in quantum communication \cite{briegel1998quantum} and quantum information processing \cite{kimble2008quantum,sipahigil2016integrated,lvovsky2009optical}. One current research direction to develop quantum memories focuses on increasing the excitation storage times for qubits coupled to waveguides \cite{roy2017colloquium,calajo2019exciting}. While a single qubit coupled to a waveguide decays spontaneously and cannot be utilized as an efficient excitation storage by itself, special arrangements of multiple qubits can lead to destructive interference that suppresses spontaneous emission---a phenomenon called \emph{subradiance} \cite{begzjav2019permutation,albrecht2019subradiant,kornovan2019extremely,facchi2016bound,tufarelli2013dynamics,asenjo2017exponential,zhang2019theory}, which leads to enhanced storage time of excitations in multi-qubits systems \cite{scully2015single}. 

One complication of utilizing subradiance to enhance storage times comes from the fact that subradiant states are usually accompanied by a superradiant state (both are a result of interference) \cite{dicke1954coherence}. It is thus important to isolate the conditions that guarantee exclusive subradiant state preparation. One such scheme, known as symmetry-protected single-photon subradiance, has been proposed for $N$ atoms in free space with arbitrarily small distances \cite{cai2016symmetry} to aid quantum memory applications. However, practical applicability of such a structure is still unclear, and the authors do not provide a definitive answer to the question whether symmetry-protection leads to optimal storage times. Here, we provide a complete discussion of this concept for waveguide QED \cite{song2018photon,das2018photon,ruostekoski2017arrays,liao2015single}, where subradiance effects persist for non-zero qubit separations \cite{zhou2017single,tsoi2008quantum,dinc2019exact,facchi2019bound}.

In this paper, we investigate whether symmetry-protected subradiance occurs in an array of microscopically separated qubits coupled to a one-dimensional waveguide, and if so, whether it is the most optimal strategy. We show that symmetry-protected subradiance indeed exists in waveguide QED, but can be sub-optimal. Then, we provide the optimal strategy, which lies in finding the eigenvalues and eigenvectors of the coupling matrix. Finally, we conclude with remarks on future directions by discussing the non-Markovian regime \cite{fang2018non,carmele2020pronounced,grimsmo2015time,pichler2017universal,pichler2016photonic,dinc2019exact}, where between-qubit photon propagation time delay can no longer be ignored.

Our interest lies in the collective spontaneous emission from a linear chain of identical $N$ qubits, equally separated by a distance $L$, where each qubit has  energy gap $\Omega$. Spontaneous emission dynamics  can be derived from the collective decay rates of the entire system \cite{dinc2019exact}. To compute these, it is sufficient to consider only the single-excitation subspace. Let us start by writing down a general single-excitation state:
\begin{equation} \label{eq:generalized}
\begin{split}
        \ket{\psi(t)}&=\int_{-\infty}^\infty \diff x  [\psi_L(x,t)C_L^\dag(x)+\ldots\\&+\psi_R(x,t)C_R^\dag(x)] \ket{0} + \sum_{Q=1}^N \alpha_Q(t) e^{-i\Omega t} \ket{e_Q}.
\end{split}
\end{equation}
Here, $C_{R/L}^\dag(x)$ is the photon creation operator, $\psi_{R/L}(x,t)$ is the field amplitude for the right/left moving photons, $\alpha_Q(t)$ is the excitation coefficient for the qubit $Q$. $\ket 0$ is the vacuum state and $\ket{e_Q}$ is the state where only the qubit $Q$ is excited. 

Now, for our purposes, we assume that the qubits are microscopically separated such that inter-qubit propagation time-delays are neglected (Markovian regime). Moreover, we assume that the field is initially in the vacuum state. Then, the time evolution of qubit excitation coefficients is governed by a differential equation \cite{dinc2019exact,dinc2020diagrammatic}:
\begin{equation}\label{eq:timeevol}
   \dot \alpha_j(t) = -\frac{\gamma_0}{2} \sum_l J_{jl} \alpha_l(t),
\end{equation}
where $J_{jl}= e^{i|j-l|\theta}$ is the dimensionless qubit-qubit coupling matrix with $\theta=\Omega L$ and $\gamma_0$ is the single-emitter decay rate. 

For our purposes, it suffices to consider only the coupling matrix $J$
\begin{equation}
    J \dot = 
    \begin{bmatrix}
    1 & e^{i\theta} & \ldots & e^{i \theta (N-1)} \\
    e^{i\theta} & 1 &\ldots & e^{i \theta (N-2)} \\
    . & . & . & . \\
    . & . & . & . \\
    . & . & . & . \\
    e^{i\theta(N-1)} & e^{i \theta (N-2)} & \ldots & 1 
    \end{bmatrix},
\end{equation}
since the eigenvalues and eigenvectors of $J$ define the collective decay rates and interactive basis of the collective system. Large eigenvalues, e.g. decay rates, correspond to faster decay modes, whereas small eigenvalues signal decay modes that keep the excitation in the system for longer times. Mathematically, subradiance corresponds to the case where the coupling matrix is nearly singular. The subradiant states become bound-states in continuum (BIC) when the coupling matrix becomes exactly singular.

Before describing symmetry-protected subradiance, we first discuss symmetric and anti-symmetric states and their time-evolution. Symmetric (anti-symmetric) states are those that have $+1$ ($-1$) eigenvalue w.r.t. the mirror operator $\hat P$ that we discuss in SM 1. One can think of $\hat P$ as a generalized version of the parity operator from ordinary quantum mechanics. In \cite{dinc2019exact}, we conjectured the subdivision of symmetric and anti-symmetric decay modes for a linear chain of qubits, where an initially prepared symmetric (anti-symmetric) state couples to only symmetric (anti-symmetric) decay modes. Now, we present the proof.
\begin{theorem}[The symmetric and anti-symmetric collective decay rates]
The symmetric (anti-symmetric) states, that acquire a pre-factor of $\pm1$ upon mirroring with respect to the center, couple to only symmetric (anti-symmetric) collective decay rates. For even $N$, the dimensionality of symmetric and anti-symmetric subspace is equal to $N/2$. For odd $N$, the symmetric subspace is larger than the anti-symmetric subspace by a single dimension. 
\end{theorem}
\begin{proof}
See SM 1. 
\end{proof}

Theorem 1 guarantees that a symmetric or anti-symmetric preparation of the initial state eliminates nearly half of the decay modes. Such a property can be utilized for quantum memory applications, if the superradiant decay mode can be selectively and consistently eliminated even in the presence of experimental imperfections and parametric uncertainties. If this is possible, all the remaining modes would be subradiant, hence with long decay times.

Now, since subradiance emerges when $J$ becomes nearly singular, we find this condition:
\begin{theorem}[Singularity condition]
The coupling matrix is singular only for the discrete values $\theta=n\pi$.
\end{theorem}
\begin{proof}
See SM 2.
\end{proof}

The singularity condition obtained from the coupling matrix agrees with the previous results obtained from the propagators in the single-atomic-excitation subspace \cite{facchi2019bound}. Physically, the singularity condition refers to the case, where out of $N$ collective decay rates, $N-1$ become zero (extremely subradiant) and one becomes superradiant \cite{dinc2019exact}. Such a condition would be perfect for quantum memory applications, as an initially prepared state would not decay even for long time intervals. However, due to the discrete nature of the singularity condition, it is improbable to experimentally obtain the exact condition. We usually observe a less extreme subradiance (for $\theta \approx n \pi$) where $N-1$ decay rates are small but non-zero. In fact, for a linear chain of $N$ qubits, the most subradiant decay rate decreases with increasing $N$ \cite{albrecht2019subradiant,zhang2019theory}. Therefore, one approach for increasing the memory application potential of a quantum system is to increase the number of qubits, which leads to the natural question: In which superposition should the qubits be excited? 

One potential answer to this question lies in the symmetry-protected subradiance. Symmetry-protection guarantees that anti-symmetry leads to subradiance  free space. In waveguide QED, the distinction is more subtle, as the subradiance condition is not only $\theta\approx0$ (as it is in free space), but extends to a countable infinite set of points $\theta\approx n\pi$. Now, we state an equivalent theorem for waveguide QED:
\begin{theorem}[Symmetry-protected subradiance]
For $\theta\approx 2 n \pi $, the superradiant state is symmetric and anti-symmetric states are guaranteed to be subradiant. For $\theta \approx (2n+1)\pi$, superradiant state is  symmetric (anti-symmetric) for odd (even) qubit number $N$.
\end{theorem}
\begin{proof}
See SM 3.
\end{proof}

Unlike in free space, the properties of the superradiant state depend on the specific system geometry. For a linear chain of $N$ qubits, the superradiant state can be given as (SM 3)
\begin{equation}
    \ket{\psi_{\rm sup}}=\sum_{j=1}^N (-1)^{nj} \ket{e_j}.
\end{equation}
If $n$ is even, then $\ket{\psi_{\rm sup}}$ becomes the Dicke state. For odd $n$, $\ket{\psi_{\rm sup}}$ has alternating signs. Depending on $N$ and $\theta$ (or $n$ in $\theta=n\pi$), superradiant state can be either symmetric or anti-symmetric. Symmetry-protection of subradiance guarantees that the opposite subspace is always subradiant.

While the symmetry-protection in waveguide QED can guarantee subradiance, we have found that it might lead to sub-optimal subradiance. More concretely, there may be subradiant states that decay extremely slowly but have the same symmetric properties as the superradiant state. In the following, we provide a detailed discussion based on a simple example.

Consider the case $N=3$ and $\theta = 2\pi +\delta$, where $\delta$ is a small parameter. The decay rates can be given as \cite{dinc2019exact}
\begin{subequations} 
\begin{align}
      \Gamma_{\rm sup}^{(+)}(\delta) &= 3 + O(\delta), \\
    \Gamma_{\rm sub}^{(+)}(\delta) &= - \frac{2}{3}i \delta + \frac{2}{27}\delta^2 + O(\delta^3) , \\
    \Gamma_{\rm sub}^{(-)}(\delta) &=   -2i \delta +2 \delta^2 + O(\delta^3),
\end{align}
\end{subequations}
with corresponding eigenmodes:
\begin{subequations} 
\begin{align}
      |\psi_{\rm sup}^{(+)}(\delta)\rangle &\approx \frac{1}{\sqrt{3}}\left[\ket{e_1}+\left(1-\frac{i\delta}{3}\right)\ket{e_2}+\ket{e_3}\right]   ,\\
     |\psi_{\rm sub}^{(+)}(\delta)\rangle &\approx \frac{1}{\sqrt{6}}\left[\ket{e_1}-\left(2+\frac{2i\delta}{3}\right) \ket{e_2}+\ket{e_3}\right]  , \\
     |\psi_{\rm sub}^{(-)}(\delta)\rangle &= \frac{1}{\sqrt{2}}[\ket{e_1}-\ket{e_2}].
\end{align}
\end{subequations}
Here, we set $\gamma_0=1$ for simplicity, consider the Taylor expansion up to $O(\delta^2)$, and the subscripts ``sup''/``sub'' refer to superradiance/subradiance. While the symmetry-protected subradiance guarantees that $\Gamma_{\rm sub}^{(-)}$ is subradiant, $\Gamma_{\rm sub}^{(+)}$ has a $27$-fold smaller real part for a leading term in $\delta^2$. Thus, the symmetric subradiant state decays with a $27$-fold smaller exponential factor than the anti-symmetric one. 

Now, consider the time evolution of an initially excited state $\ket{\rm{sub^{(+)}}}=|\psi_{\rm sub}^{(+)}(0)\rangle=\frac{1}{\sqrt{6}}[\ket{e_1}-2\ket{e_2}+\ket{e_3}]$ for a finite, but small imperfection $\delta$. Due to this imperfection, there is a chance that this initial preparation decays partially through the superradiant decay mode. Fortunately, this portion is negligible compared to the portion decaying through the subradiant portion, since
\begin{equation}
    \begin{split}
       |\rm{sub^{(+)}}\rangle  &\approx \frac{2i\delta}{9\sqrt{2}}  |\psi_{\rm sup}^{(+)}(\delta)\rangle + \left(1- \frac{2i\delta}{9}  \right)|\psi_{\rm sub}^{(+)}(\delta)\rangle, \\
       &\approx O(\delta) |\psi_{\rm sup}^{(+)}(\delta)\rangle + O(1) |\psi_{\rm sub}^{(+)}(\delta)\rangle.
    \end{split}
\end{equation}
This means that $O(\delta)$ portion of $\ket{\rm sub^{(+)}}$ decays with a superradiant decay rate, whereas $O(1)$ portion decays with the optimal subradiant decay rate. For all practical purposes, $\ket{\rm sup^{(+)}}$ has a nearly zero superradiant part and is therefore subradiant even for non-zero, but small, deviations from the perfect condition, i.e. $\delta \neq 0$. Fig. \ref{fig:fig1} illustrates this concept for a particular example with $N=3$ qubits and $\theta = 2.1\pi$ with $\delta=0.1\pi$. In this figure, we calculate the qubit excitation probabilities following the approach described by \cite{dinc2019exact} in the Markovian regime. Thus, while the symmetry-protection provides a fully subradiant behavior, a clever preparation of the initial state can provide a stronger subradiance behavior, thus a better potential for quantum memory applications.

\begin{figure}
    \centering
    \includegraphics[width=\columnwidth]{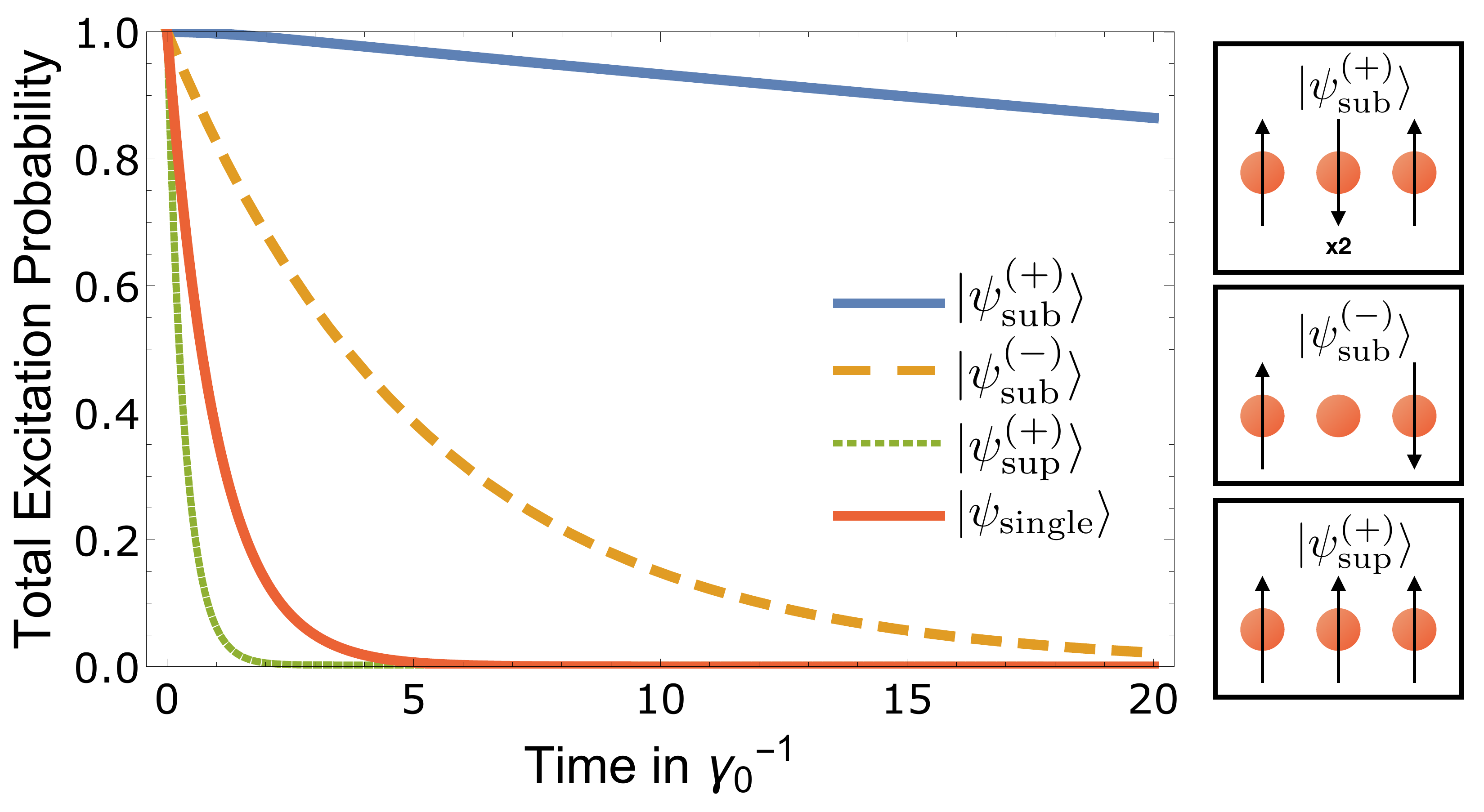}
    \caption{The total qubit excitation probabilities (left) and state preparation conditions with the upward/downward arrows representing positive/negative coherent excitations (right). We illustrate the decay of three initially prepared states for $N=3$ and of a single-qubit for reference. Here, we define $|\psi\rangle=|\psi(\delta=0)\rangle$ and pick $\theta=2\pi + 0.1 \pi$. While the symmetry protection guarantees subradiance for the anti-symmetric state, the symmetric subradiant state has a lower decay, even in the presence of parametric imperfections such as $\delta=0.1\pi$.}
    \label{fig:fig1}
\end{figure}

So far in this paper, we have considered the case where the qubit separation is small enough such that the photon-propagation time between adjacent qubits is negligible. Now, we argue that consideration of time-delayed quantum feedback within the system is the natural next step and a necessity for the experimental realization of long storage times for qubits.

The idea that $N$ qubits have $N$ collective decay rates applies only for the Markovian regime. As the time delayed feedback caused by the photon propagation between the qubits becomes prominent, non-Markovian processes become more significant and lead to infinitely many non-Markovian decay rates. The question is: How do these non-Markovian processes affect the subradiance behavior of the qubits? 

Large qubit separation leads to fully subradiant behavior of the system \cite{zheng2013persistent,dinc2019exact}, but are more susceptible to propagation losses and unwanted non-radiative decay. Compact multi-qubit systems with small qubit separation have lower propagation losses inside the waveguide and provide higher probability of exciting BIC through vacuum decay \cite{calajo2019exciting}. On the other hand, small separation would be hard to implement experimentally, might suffer from unwanted dipole-dipole interactions \cite{cheng2017waveguide} and lead to sub-optimal BIC excitation through time-delayed feedback and multi-photon scattering \cite{calajo2019exciting}. Since the subradiant decay couples very slightly to the waveguide, state preparation via single-photon states would be experimentally implausible. Therefore, BIC excitation through time-delayed feedback and multi-photon scattering is indeed needed to prepare the initial state, although it is an open question how to excite a specific subradiant state beyond $N=2$ qubits \cite{calajo2019exciting}. Consequently, the most optimal qubit separation to enhance BIC generation and to minimize non-radiative losses, which is also experimentally feasible, turns out to be moderate separation. This regime provides a higher compatibility, gives an experimental freedom to separate qubits sufficiently to avoid direct dipole-dipole interactions \cite{cheng2017waveguide} and shows a stronger subradiant and superradiant behavior than the Markovian regime \cite{dinc2019exact}. 

Unfortunately, the coupling matrix approach we have taken so far allows us to draw conclusions only for the Markovian regime, as for the non-Markovian regime, the time-evolution given by Eq. (\ref{eq:timeevol}) is no longer valid. In the non-Markovian regime, the eigenspectrum of the coupling matrix can no longer provide the optimal subradiance, but can provide approximations. 

Moreover, we have seen that symmetry-protection may lead to sub-optimal subradiance in the Markovian regime. Even if symmetry-protection exists in the non-Markovian regime, sub-optimality would still be an issue due to the continuity of the decay rates, e.g. as $\theta$ reaches zero, the non-Markovian decay rates coincide with the Markovian decay rates \cite{dinc2019exact,zheng2013persistent}. Therefore, while our discussion of the symmetry-protected and optimal subradiance may aid the non-Markovian explorations, further research is required to reach conclusions for this regime. Until then, one can seed numerical optimization algorithms with the (Markovian) eigenstates of the coupling matrix and obtain locally optimal states that decay with a subradiant decay rate, or for small number of qubits, use the diagrammatic non-Markovian time evolution approach developed in \cite{dinc2020diagrammatic} to obtain analytical solutions for the optimal subradiant decay.

In this paper, we have discussed the practical use of subradiance for increased excitation storage times starting from the time evolution of a general single-excitation state. We have proven the symmetric and anti-symmetric collective decay rates theorem, first conjectured in \cite{dinc2019exact} and used it to state a symmetry-protected subradiance theorem in waveguide QED equivalent to the one for the free space \cite{cai2016symmetry}. Then, we have shown that symmetry-protection, while guarantees subradiance, can lead to sub-optimal subradiance in waveguide QED. Finally, we discussed natural future directions of our work by discussing a more experimentally relevant regime where non-Markovian effects become prominent. We expect our work to be beneficial towards designing memory applications for future quantum technologies.

\emph{Acknowledgements} --- Research at Perimeter Institute is supported in part by the Government of Canada through the Department of Innovation, Science and Economic Development Canada and by the Province of Ontario through the Ministry of Colleges and Universities.

\newpage
\onecolumngrid
\setcounter{equation}{0}
\setcounter{theorem}{0}
\renewcommand{\theequation}{S\arabic{equation}}
\renewcommand{\thefigure}{S\arabic{figure}}
\renewcommand{\bibnumfmt}[1]{[S#1]}
\renewcommand{\thesection}{SM\arabic{section}}
\section{Proof of theorem 1: generalized mirror operator} 
In this section, we will provide the proof of the theorem 1 presented in the main text:
\begin{theorem}[The symmetric and anti-symmetric collective decay rates]
The symmetric (anti-symmetric) states, that acquire a pre-factor of $\pm1$ upon mirroring with respect to the center, couple to only symmetric (anti-symmetric) collective decay rates. For even $N$, the dimensionality of symmetric and anti-symmetric subspace is equal to $N/2$. For odd $N$, the symmetric subspace is larger than the anti-symmetric subspace by $1$. 
\end{theorem}

Let us start by recalling the mirror operator, $\hat P$, for the single-qubit excitation space \cite{dinc2019exact}:
\begin{equation}
    \hat P \ket{e_j} = \ket{e_{N-j+1}},
\end{equation}
where $j$ is a dummy variable ranging from $1$ to $N$ denoting the qubit identity. On the qubit excitation subspace, we represent $\ket{e_j}$ as the unit vectors of an $N$-dimensional complex space. In this space, the points represent the coherent superposition of qubits. Then, the matrix representation for the mirror operator is as follows:
\begin{equation}
    P \dot =
    \begin{bmatrix}
   0 &  \ldots &  0 &1 \\
   0 &  \ldots &  1 &0 \\
    . & . & . & . \\
    1 &  \ldots &  0 &0 
    \end{bmatrix}
\end{equation}
The action of this matrix on the $j$th unit vector swaps the excitations of qubits $N-j+1$ and $j$. 
\begin{claim}
The eigenvalues of the mirror operator are $\pm 1$, where $+1$ corresponds to symmetric states and $-1$ corresponds to anti-symmetric states.
\end{claim}
The mathematical definition of symmetry and anti-symmetry comes from the mirror operator. Intuitively, one can consider the mirror operator as flipping the qubits with respect to the center. If a state is excited in a symmetric manner, it will remain unchanged. If it is anti-symmetric, then it will acquire a $-1$ pre-factor. Then, we can prove this claim by picking the eigenstates as the symmetric and anti-symmetric unit vectors and showing that they are indeed eigenstates. Once we find all $N$ eigenstates, the proof is over.

For even $N=2M$, let us pick the set of eigenvectors as
\begin{equation}
    \ket{\xi} \sim \frac{1}{\sqrt{2}} \left( \ket{e_j} \pm \ket{e_{N-j+1}} \right), \quad \forall j = 1, \ldots, M.
\end{equation}
There are $N/2$ symmetric and $N/2$ anti-symmetric eigenstates. For odd $N=2M+1$, we have an additional eigenstate:
\begin{equation}
    \ket{\xi} \sim \ket{e_{(N+1)/2}},
\end{equation}
which is a symmetric eigenstate. Hence, we have found $N$ distinct eigenstates for the mirror operator, which have only $\pm 1$ eigenvalues.
\begin{claim}
The eigenstates of the coupling matrix can be chosen to be either symmetric or anti-symmetric.
\end{claim}
To prove this claim, we need to show that the coupling matrix and the mirror matrix commute such that $[J,P]=0$. Denoting $P_{ij}=\delta_{i (N+1-j)}$, where $\delta$ is the Kronecker delta, we find
\begin{equation}
\begin{split}
    (J P)_{ak} &= \sum_j J_{aj} \delta_{j (N+1-k)} = J_{a (N+1-k)} = e^{i\theta|N+1-k-a|}, \\
    (P J)_{ak} &= \sum_j \delta_{a (N+1-j)} J_{jk} = J_{(N+1-a)k} =e^{i\theta|N+1-k-a|},
\end{split}
\end{equation}
where $JP = PJ \implies [J,P]=0$. Thus, a general time-evolution solution (where we define $\vec \alpha = [\alpha_1, \ldots, \alpha_N]$) becomes:
\begin{equation}
    \vec \alpha = \sum_{k=1}^N x_k \vec \xi^{(k)} e^{-\lambda_k t},
\end{equation}
where $\vec \xi^{(k)}$ is the $k$th eigenvector, $x_k$ are some coefficients, and $\lambda_k$ is the corresponding eigenvalue. Here, $\vec \xi^{(k)}$s are either symmetric or anti-symmetric. Thus, for an initial symmetric excitation $\vec \alpha_0 = \sum_k x_k \vec \xi^{(k)}$, $x_k=0$ for when $\xi^{(k)}$ is anti-symmetric. This is because $\xi^{(k)}$ are linearly independent and a re-summation of all $\sum_{k \in \text{anti-sym}} x_k \vec \xi
^{(k)}=0$, which is a requirement for $\vec \alpha_0$ to be symmetric, is only zero when all such $x_k=0$. A similar argument can be made for an initial anti-symmetric $\vec \alpha_0$. This finishes the main body of the proof.

Now, we can bring the proof together by summarizing the main steps. The mirror operator has $N$ eigenvalues, which are either $+1$ or $-1$. The degeneracy of the eigenvalues depend whether $N$ is even or odd. If $N$ is even, there are $N/2$ of each. For odd $N$, the symmetric subspace is larger by a one dimension. The coupling matrix commutes with the mirror matrix, hence any eigenvector of the coupling matrix is either symmetric or anti-symmetric (or can be chosen to be, if there is a degeneracy). Since the Hamiltonian commutes with the mirror operator \cite{dinc2019exact}, the symmetric properties of an initially prepared state does not change during time-evolution. Hence, an initially symmetric/anti-symmetric prepared state will decay only through symmetric/anti-symmetric decay modes, which we have shown analytically by writing down a general solution of the time-evolved state.

Throughout this proof (and thereafter), we make the assumption that $J$ has a non-degenerate spectrum. This is indeed the case for many examples we have considered for various $N$ and $\theta \neq n\pi$ and where the decay mode description is reasonable. For cases where $J$ is non-degenerate (or even non-diagonalizable, which we haven't encountered so far), for example when $\theta=n \pi$, we can use some small perturbations $\epsilon$ and take the limit $\epsilon\to 0$, defining the symmetric and anti-symmetric subspaces asymptotically.

\section{Proof of theorem 2: LU decomposition of the coupling matrix} 
In this section, we will provide the proof of the theorem 2 in the main text:
\begin{theorem}[Singularity condition]
The coupling matrix is singular only for the discrete values $\theta=n\pi$.
\end{theorem}

We start by considering the lower and upper triangular decomposition of the coupling matrix.
\begin{claim}
The coupling matrix can be decomposed into a lower and upper triangular form $J=LU$, where $U$ is singular if and only if $\theta=n\pi$ ($n$ is a non-negative integer) and $L$ is non-singular.
\end{claim}
To prove this claim, we first state the $L$ and $U$ matrices and show that $J=LU$. We start by defining a complex number p for generalization:
\begin{equation}
    p \dot = e^{i \theta}
\end{equation}
Matrix form of the $N$-dimensional lower triangular matrix $L$ is as follows:
\begin{equation}
    L \dot =
    \begin{bmatrix}
   1 & 0 &  \ldots  & 0  \\
   p &  1 &  \ldots  & 0 \\
   . & . & . & . \\
    p^{N-1} & p^{N-2}  & \ldots  &1 
    \end{bmatrix}, \quad \quad
    L_{aj} = 
    \begin{cases}
        0 & a<j \\
        p^{a-j} & a \geq j
    \end{cases}
\end{equation}
Matrix form of the $N$-dimensional upper triangular matrix $U$ is as follows:
\begin{equation}
        U \dot = 
    \begin{bmatrix}
        1 & p & p^2 & p^3 & \dots  & p^{N-1} \\
        0 & 1-p^2 & p(1-p^2) & p^2(1-p^2) & \dots  & p^{N-2}(1-p^2) \\
        0 & 0 & 1-p^2 & p(1-p^2) & \dots  & p^{N-3}(1-p^2) \\
         \vdots &   \vdots &   \vdots &   \vdots &  \vdots &   \vdots \\
        0 & 0 & 0 & 0 & \dots  & 1-p^2
    \end{bmatrix}, \quad \quad
        U_{aj} = 
    \begin{cases}
        p^{j-1} & a=1 \\
        p^{j-a}(1-p^2) &  1 < a \leq j \\
        0  & j < a
    \end{cases}
\end{equation}
We can find $LU$ by direct matrix multiplication.
\begin{equation}
    (LU)_{aj} = \sum_{k=1}^{N} L_{ak}U_{kj} 
\end{equation}

We divide matrix multiplication into two separate cases. For $a \leq j$:
\begin{equation}
     (LU)_{aj} = L_{a1} U_{1j} + \sum_{k=2}^{a} L_{ak}U_{kj} + \sum_{k=a+1}^{N} L_{ak}U_{kj} 
\end{equation}
For the final summation $k > a$, then $L_{ak} = 0$ for this summation. Then $(LU)_{aj}$ simplifies to:
\begin{equation}
   (LU)_{aj} = L_{a1} U_{1j} + \sum_{k=2}^{a} L_{ak}U_{kj} 
\end{equation}
For $2 \leq k \leq a$ plug in values of matrices which are $L_{ak} = p^{a-k}$ and $U_{kj} = p^{k-j}(1-p^2)$ then we obtain the following equation:
\begin{equation}
    (LU)_{aj} = p^{a+j-2} + p^{a+j} \sum_{k=2}^{a} (p^{-2k} - p^{-2k+2}) = p^{j-a}
\end{equation}

For $j<a$, we follow similar steps and obtain:
\begin{equation}
   (LU)_{aj} = L_{a1} U_{1j} + \sum_{k=2}^{j} L_{ak}U_{kj} =p^{a-j}.
\end{equation}
By considering these two cases together, the final form of the $LU$ is as follows:
\begin{equation}
    (LU)_{aj} = p^{|a-j|} = J_{aj}
\end{equation}

Now, we return to the discussion of singularity. $L$ and $U$ are triangular matrices, then determinant of these matrices are equal to the multiplication of the elements on the main diagonal. Therefore, $\det(L) = 1$ and $\det(U) = (1-p^2)^{N-1}$. Using $\det(J)=\det(L)\det(U)$, we have that $\det(J) = (1-p^2)^{N-1}$. Therefore, $J$ is singular ($\det(J)=0$) if and only if $p= \pm 1$ which is equivalent to $\theta = n \pi$ ($n$ is a non-negative integer). This finishes the proof.

\section{Proof of theorem 3: symmetry-protection in waveguide QED} 
In this section, we will provide the proof of the theorem 3 in the main text:
\begin{theorem}[Symmetry-protected subradiance]
For $\theta\approx 2 n \pi $, the superradiant state is symmetric and anti-symmetric states are guaranteed to be subradiant. For $\theta \approx (2n+1)\pi$, superradiant state is  symmetric (anti-symmetric) for odd (even) qubit number $N$.
\end{theorem}

Let us start with $\theta = 2n \pi + \delta$, where $\delta$ is a small parameter such that 
\begin{equation}
    e^{i\theta} = e^{i(2n \pi + \delta)} = e^{i \delta} = 1+ i\delta + O(\delta ^ 2)
\end{equation}
Coupling matrix can be written as  below after neglecting higher order terms:
\begin{equation}
    J \approx \begin{bmatrix}
    1 & 1 & \ldots & 1 \\
    1 & 1 &\ldots & 1 \\
    . & . & . & . \\
   1 & 1 & \ldots & 1 
    \end{bmatrix}
    + i \delta
    \begin{bmatrix}
    0 & 1 & \ldots & N-1 \\
    1 & 0 &\ldots & N-2 \\
    . & . & . & . \\
   N-1 & N-2 & \ldots & 0 
    \end{bmatrix}
\end{equation}

Define two new matrices such that $J \approx J_0 +i \delta M$,
\begin{equation}
   J_0 \dot = \begin{bmatrix}
    1 & 1 & \ldots & 1 \\
    1 & 1 &\ldots & 1 \\
    . & . & . & . \\
   1 & 1 & \ldots & 1 
    \end{bmatrix}
\end{equation}

\begin{equation}
    M \dot = 
    \begin{bmatrix}
    0 & 1 & \ldots & N-1 \\
    1 & 0 &\ldots & N-2 \\
    . & . & . & . \\
   N-1 & N-2 & \ldots & 0 
    \end{bmatrix}
\end{equation}

$J_0$ has only one non-zero eigenvalue which is corresponds to superradiant state. $\lambda_{\rm sup} = N$ is the eigenvalue and $w_{\rm sup} = \begin{bmatrix} 1, \ldots ,1 \end{bmatrix}^T$ is the corresponding eigenvector. In addition, $w_{\rm sup}^T J_0 = \lambda_{\rm sup} w_{\rm sup}^T$. For a sufficiently small $\delta$, the eigenvalues and eigenvectors of $J$ can be approximated from the known eigenvalues and eigenvectors of the $J_0$. Let $\lambda_J$ and $w_J$ be the corresponding eigenvalue and eigenvector of the $J$. Then we can approximate $\lambda_J$ as:
\begin{equation}
    \lambda_J = \lambda_{\rm sup} + i \delta \dfrac{w_{\rm sup}^T M w_{\rm sup}}{w_{\rm sup}^Tw_{\rm sup}} + O(\delta^2),
\end{equation}
where we realize that $\text{Re}[\lambda_J]=N+ O(\delta^2)$ the decay portion of the eigenvalue vanishes to the first order. For the eigenvector, we pick $w_J = w_{\rm sup}+O(\delta)$. Hence, for small $\delta$, the Dicke state is the superradiant state:
\begin{equation}
    \ket{\psi_{\rm sup}}=\frac{1}{\sqrt{N}} \sum_{j=1}^N \ket{e_j}.
\end{equation}
This is a symmetric eigenstate and corresponds to a symmetric $\vec \xi^{(1)}$. Thus, for an anti-symmetric initial condition $\vec \alpha_0$, corresponding coefficient for the superradiant decay mode is zero up to $O(\delta)$. Thus, for $\theta \approx 2 \pi n$, the anti-symmetric coherent excitation subspace is guaranteed to be subradiant.

A similar calculation can be performed for $\theta = (2n+1)\pi + \delta$, where we obtain that $\text{Re}[\lambda_J]=N + O(\delta^2)$ and $w_J = w_{\rm sup} + O(\delta)$ with the only difference that $w_J=[-1,1,-1,\ldots]$. Then, the superradiant state becomes:
\begin{equation}
    \ket{\psi_{\rm sup}}=\frac{1}{\sqrt{N}} \sum_{j=1}^N (-1)^j \ket{e_j}.
\end{equation}
Now, this state is symmetric if $N$ is odd and anti-symmetric if $N$ is even. Thus, the symmetric (anti-symmetric) subspace is guaranteed to be subradiant for even (odd) $N$.

\end{document}